\documentclass[11pt]{article}
\usepackage{amsfonts}
\usepackage{epsfig}
\usepackage{graphicx}
\usepackage{float}
\usepackage{latexsym}
\usepackage{amssymb}
\usepackage{amsmath}
\usepackage{amsthm}

\newtheorem{fait}{Fact}[section]
\newtheorem{theoreme}[fait]{Theorem}
\newtheorem{lemme}[fait]{Lemma}

\newtheorem{proposition}[fait]{Proposition}
\newtheorem{Cor}[fait]{Corollary}

\newcommand{\Ep}{\mathbb{E}}
\newcommand{\Vp}{\mathbb{V}ar}
\newcommand{\Pp}{\mathbb{P}}
\newcommand{\R}{\mathbb{R}}
\newcommand{\Z}{\mathbb{Z}}
\newcommand{\N}{\mathbb{N}}

\newcommand{\dg}{\Delta g}

\renewcommand{\tt}{S}

\newcommand{\un}{1\!\!1}

\title{DNA unzipping via stopped birth and death processes with unknown transition probabilities
\author{P. Andreoletti, R. Diel \footnote{Laboratoire MAPMO - C.N.R.S. UMR 6628 - F\'ed\'eration Denis-Poisson, Universit\'e d'Orl\'eans, 
(Orl\'eans France). \newline \vspace{0.1cm}  $\quad$  MSC 2000 62P10 ; 82D30. \newline \vspace{0.5cm} \textit{Key words :  DNA unzipping, birth and death processes, random environment, maximum of likelihood} }
}}

\begin{document}
\maketitle
\begin{abstract}
 In this paper we provide an alternative approach to the works of the physicists S. Cocco and R. Monasson about a model of DNA molecules. The aim is to predict the sequence of bases by mechanical stimulations. The model described by the physicists is a stopped birth and death process with unknown transition probabilities. We consider two models, a discrete in time and a continuous in time, as general as possible. We show that explicit formula can be obtained for the probability to be wrong for a given estimator, and apply it to evaluate the quality of the prediction. Also we add some generalizations comparing to the initial model allowing us to answer some questions asked by the physicists.
 \end{abstract}

\section{Introduction}

\subsection{The physical approach}

In this introduction we first summarize some ideas and results of the works of V. Baldazzi, S. Cocco, E. Marinari and R. Monasson  (\cite{Monasson0}, \cite{Monasson1}), and S. Cocco and R. Monasson \cite{Monasson2} who are interested in a method for DNA molecules sequencing. They study a mechanical way, described below, instead of  traditional bio-chemical or gel electrophoresis technics. The experiments for mechanical unzipping were first realized by Bockelmann, Helsot and coworkers  \cite{Bockelmann} and \cite{Bockelmann2}. The principle is based on the fact that the link strength between two bases of a given pair depends on whether it is a $C \equiv G$ or a $A-T$ (see Figure \ref{fig1}).
\begin{figure}[ht]
\begin{center}
\begin{picture}(0,0)%
\includegraphics{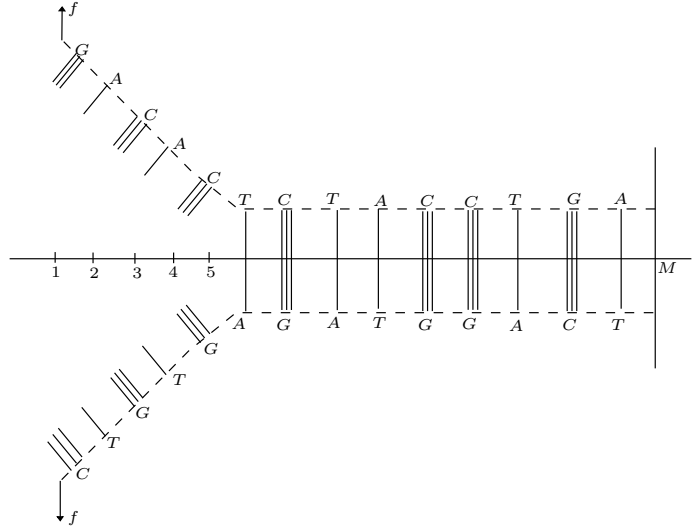}%
\end{picture}%
\setlength{\unitlength}{3947sp}%
\begingroup\makeatletter\ifx\SetFigFont\undefined%
\gdef\SetFigFont#1#2#3#4#5{%
  \reset@font\fontsize{#1}{#2pt}%
  \fontfamily{#3}\fontseries{#4}\fontshape{#5}%
  \selectfont}%
\fi\endgroup%
\begin{picture}(4306,3330)(3657,-4636)
\put(7447,-3384){\makebox(0,0)[lb]{\smash{{\SetFigFont{6}{7.2}{\rmdefault}{\mddefault}{\updefault}$T$}}}}
\put(7139,-3390){\makebox(0,0)[lb]{\smash{{\SetFigFont{6}{7.2}{\rmdefault}{\mddefault}{\updefault}$C$}}}}
\put(4440,-3057){\makebox(0,0)[lb]{\smash{{\SetFigFont{6}{7.2}{\rmdefault}{\mddefault}{\updefault}3}}}}
\put(3928,-3054){\makebox(0,0)[lb]{\smash{{\SetFigFont{6}{7.2}{\rmdefault}{\mddefault}{\updefault}1}}}}
\put(4170,-3058){\makebox(0,0)[lb]{\smash{{\SetFigFont{6}{7.2}{\rmdefault}{\mddefault}{\updefault}2}}}}
\put(4667,-3054){\makebox(0,0)[lb]{\smash{{\SetFigFont{6}{7.2}{\rmdefault}{\mddefault}{\updefault}4}}}}
\put(4902,-3051){\makebox(0,0)[lb]{\smash{{\SetFigFont{6}{7.2}{\rmdefault}{\mddefault}{\updefault}5}}}}
\put(5349,-2599){\makebox(0,0)[lb]{\smash{{\SetFigFont{6}{7.2}{\rmdefault}{\mddefault}{\updefault}$C$}}}}
\put(5060,-3380){\makebox(0,0)[lb]{\smash{{\SetFigFont{6}{7.2}{\rmdefault}{\mddefault}{\updefault}$A$}}}}
\put(5101,-2601){\makebox(0,0)[lb]{\smash{{\SetFigFont{6}{7.2}{\rmdefault}{\mddefault}{\updefault}$T$}}}}
\put(5342,-3383){\makebox(0,0)[lb]{\smash{{\SetFigFont{6}{7.2}{\rmdefault}{\mddefault}{\updefault}$G$}}}}
\put(7734,-3025){\makebox(0,0)[lb]{\smash{{\SetFigFont{6}{7.2}{\rmdefault}{\mddefault}{\updefault}$M$}}}}
\put(4884,-3537){\makebox(0,0)[lb]{\smash{{\SetFigFont{6}{7.2}{\rmdefault}{\mddefault}{\updefault}$G$}}}}
\put(4456,-3941){\makebox(0,0)[lb]{\smash{{\SetFigFont{6}{7.2}{\rmdefault}{\mddefault}{\updefault}$G$}}}}
\put(4691,-3731){\makebox(0,0)[lb]{\smash{{\SetFigFont{6}{7.2}{\rmdefault}{\mddefault}{\updefault}$T$}}}}
\put(4276,-4127){\makebox(0,0)[lb]{\smash{{\SetFigFont{6}{7.2}{\rmdefault}{\mddefault}{\updefault}$T$}}}}
\put(4074,-1656){\makebox(0,0)[lb]{\smash{{\SetFigFont{6}{7.2}{\rmdefault}{\mddefault}{\updefault}$G$}}}}
\put(4688,-2250){\makebox(0,0)[lb]{\smash{{\SetFigFont{6}{7.2}{\rmdefault}{\mddefault}{\updefault}$A$}}}}
\put(4288,-1842){\makebox(0,0)[lb]{\smash{{\SetFigFont{6}{7.2}{\rmdefault}{\mddefault}{\updefault}$A$}}}}
\put(4509,-2067){\makebox(0,0)[lb]{\smash{{\SetFigFont{6}{7.2}{\rmdefault}{\mddefault}{\updefault}$C$}}}}
\put(4904,-2462){\makebox(0,0)[lb]{\smash{{\SetFigFont{6}{7.2}{\rmdefault}{\mddefault}{\updefault}$C$}}}}
\put(4083,-4323){\makebox(0,0)[lb]{\smash{{\SetFigFont{6}{7.2}{\rmdefault}{\mddefault}{\updefault}$C$}}}}
\put(4032,-1393){\makebox(0,0)[lb]{\smash{{\SetFigFont{6}{7.2}{\rmdefault}{\mddefault}{\updefault}$f$}}}}
\put(4032,-4594){\makebox(0,0)[lb]{\smash{{\SetFigFont{6}{7.2}{\rmdefault}{\mddefault}{\updefault}$f$}}}}
\put(5655,-2595){\makebox(0,0)[lb]{\smash{{\SetFigFont{6}{7.2}{\rmdefault}{\mddefault}{\updefault}$T$}}}}
\put(5950,-2606){\makebox(0,0)[lb]{\smash{{\SetFigFont{6}{7.2}{\rmdefault}{\mddefault}{\updefault}$A$}}}}
\put(6245,-2605){\makebox(0,0)[lb]{\smash{{\SetFigFont{6}{7.2}{\rmdefault}{\mddefault}{\updefault}$C$}}}}
\put(6521,-2614){\makebox(0,0)[lb]{\smash{{\SetFigFont{6}{7.2}{\rmdefault}{\mddefault}{\updefault}$C$}}}}
\put(6796,-2601){\makebox(0,0)[lb]{\smash{{\SetFigFont{6}{7.2}{\rmdefault}{\mddefault}{\updefault}$T$}}}}
\put(7165,-2595){\makebox(0,0)[lb]{\smash{{\SetFigFont{6}{7.2}{\rmdefault}{\mddefault}{\updefault}$G$}}}}
\put(7459,-2593){\makebox(0,0)[lb]{\smash{{\SetFigFont{6}{7.2}{\rmdefault}{\mddefault}{\updefault}$A$}}}}
\put(5656,-3375){\makebox(0,0)[lb]{\smash{{\SetFigFont{6}{7.2}{\rmdefault}{\mddefault}{\updefault}$A$}}}}
\put(5947,-3376){\makebox(0,0)[lb]{\smash{{\SetFigFont{6}{7.2}{\rmdefault}{\mddefault}{\updefault}$T$}}}}
\put(6233,-3386){\makebox(0,0)[lb]{\smash{{\SetFigFont{6}{7.2}{\rmdefault}{\mddefault}{\updefault}$G$}}}}
\put(6509,-3381){\makebox(0,0)[lb]{\smash{{\SetFigFont{6}{7.2}{\rmdefault}{\mddefault}{\updefault}$G$}}}}
\put(6802,-3393){\makebox(0,0)[lb]{\smash{{\SetFigFont{6}{7.2}{\rmdefault}{\mddefault}{\updefault}$A$}}}}
\end{picture}%
\caption{ $X_.=5$, $b_1:C\equiv G$, $b_2:A-T$} \label{fig1}
\end{center}
\end{figure}
 Indeed the link $A-T$ is weaker for biochemical reasons than the link  $C\equiv G$. Moreover, there are also some stacking effects between adjacent bases, that is to say, the force needed to break, for example, the link $C\equiv G$ is different if the $C$ is following by a $A$, or a $T$. This last factor is not negligible (see the table below) and therefore must be taken into account if we want the model to be as sharp as possible.
 \begin{figure}
 \begin{center}
\begin{tabular}{|l|l|l|l|l|}
%\hline  $i=$ & 0 & 1 \\
\hline   $g_0$ & A & T & C & G   \\
\hline   A  &  1.78 & 1.55  & 2.52 & 2.22  \\
\hline
 T  &  1.06 & 1.78  & 2.28 & 2.54  \\
\hline
 C  &  2.54 & 2.22  & 3.14 & 3.85  \\
\hline
 G  &  2.28 & 2.52  & 3.90 & 3.14  \\
\hline
\end{tabular} 
\end{center}
\caption{Binding free energies (units of $k_BT$)} \label{tab1}
\end{figure}

We now give a brief description of the experiment (for more details see \cite{Monasson0}), the extremities of the DNA molecule are stretched apart under a force $f$. The force $f$ is chosen large enough in such way that the molecule can be totally unzipped. However $f$ is also not too strong so that naturally the molecule rebuilds itself. Though there is back and forth movement of the number of open pair bases, this back and force movement generates a signal which can be measured by biologists. This signal can be modeled by a birth and death process with unknown   transition probabilities.
 
\subsection{The model} 
We denote by $M$ the length of the DNA chain and by $(b_1,b_2,$ $ \cdots,b_M)$ the sequence of bases of one of the strand of the molecule. So $b_i$ is the $i^{th}$ base which can be either a $A$, a $T$, a $C$ or a $G$ and the corresponding base of the other strand can be deduced.  We consider both a discrete and a continuous time-sequence of the number of open base pairs, the first one is denoted $X$, the second one $Y$. 
We now make the link between $X$ (and $Y$) and $b$. For this, we define the free energy $g$ of the molecule when the first $x$
base pairs are open: $$g(x):=\sum_{i=1}^xg_0(b_i,b_{i+1})-x g_{1}(f).$$
There are two different parts: first, $g_0(b_i,b_{i+1})$ is the binding energy of the pair $i$.  Note that stacking effects are taken into account: $g_0$ depends on the base content $b_i$ and on the next pair $b_{i+1}$. The second contribution $g_{1}(f)$ is the work to stretch under a force $f$ the open part of the two strands when one more base pair is opened, in particular $g_{1}$ increases when $f$ does. Note that $g_1$ is known, whereas $\sum_{i=1}^xg_0(b_i,b_{i+1})$ is unknown as we are looking for the $b_i$'s, in fact we assume that the sequence of bases is random.  A typical trajectory of $g$, obtained by numerical simulations, is given in \cite{Monasson1} page 7 and looks like Figure \ref{fig3p}.
 \begin{figure}[ht]
\begin{center}
\includegraphics[width=10cm,height=5cm]{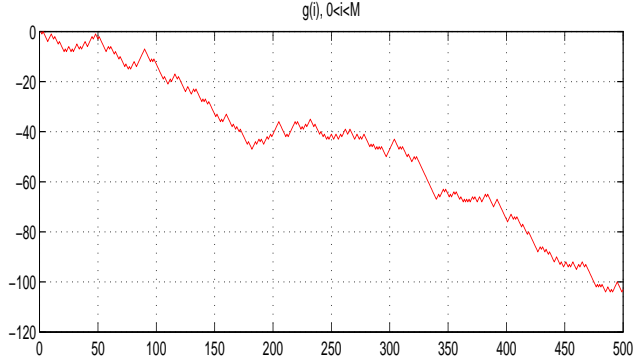}
\caption{A typical trajectory of $g$, $M=500$} \label{fig3p}
\end{center}
\end{figure}

 The number of open pairs fluctuates randomly with a distribution directly connected to the difference of the free energy $g$ between two consecutive base pairs. Therefore it can be represented by a random walk in random environment: 

\noindent  \\
\emph{ The discrete case}  is defined as follows,  assume that the random sequence $b:=(b_x,\ 1\leq x\leq M)$  is fixed, then the transition probabilities of the number of open pairs are given by: for all $1 \leq x \leq M-1$,
\begin{align} \label{loiX}
p_x=\Pp(X_{.+1}=x+1|X_.=x,b) &:= \frac{1}{1+\exp\left(\beta(g(x)-g(x-1))\right)},
\end{align}
where $\beta$ is a constant parameter which is proportional to the inverse of the temperature. Also we assume that the first base of the molecule is always open which means that $p_1=1$. Note that the larger is $f$ the greater is the probability to open a new pair. We easily get a simple expression for this probability which is
\begin{align} 
\Pp(X_{.+1}=x+1|X_.=x,\ b)&=  \frac{1}{1+\exp(\beta \dg(b_{x},b_{x+1}))}, \label{defpi}
\end{align}
where we denote 
\begin{align}
\dg(b_{x},b_{x+1}):=g_0(b_{x},b_{x+1})-g_1(f).\label{Delta}
\end{align}
 Formula (\ref{defpi}) shows that we only need local information on the sequence $b$ to get the  transition probability at site $x+1$, and that $X_.$ can only move forward with probability $p_x$ or backward with probability $1-p_x$. We discuss about some results on this well known model in the next section. A typical trajectory of $X$, obtained by  numerical simulations,  looks like Figure 4. \\

\noindent
For \emph{the continuous time model}, the physicists also take into account the time it takes $X$ to go from a site to another. Thus we introduce a second time continuous model $Y$. Given the $g_0$, when $Y$ is at the site $x$, it jumps in $x+1$ with rate $re^{-\beta g_0(b_{x},b_{x+1})}$ and in $x-1$ with rate $re^{-\beta g_1(f)}$ where $r$ is a constant which value depends on biological parameters. That is, given the DNA sequence $b$, $Y$ is a Markov process with finite state space $\{1,\dots,M\}$ killed when it hits $M$ whose transition rates are for $x\geq 2$,
$$
p(x,y)=\left\{\begin{array}{ll}
re^{-\beta g_0(b_{x},b_{y})}&\text{if }y=x+1,\\
re^{-\beta g_1(f)}&\text{if }y=x-1,\\
-r\big(e^{-\beta g_0(b_{x},b_{x+1})}+e^{-\beta g_1(f)}\big)&\text{if }y=x,\\
0 &\text{otherwise,}
\end{array}\right.
$$
and for $x=1$,
$$
p(1,y)=\left\{\begin{array}{ll}
re^{-\beta g_0(b_1,b_2)}&\text{if }y=2,\\
-re^{-\beta g_0(b_1,b_2)}&\text{if }y=1,\\
0 &\text{otherwise.}
\end{array}\right.
$$
 The process $Y$ can be represented as the couple $(X,T)$ where $X$ is the sequence of  discrete jumps and has the same law as in \eqref{loiX} and $T$ is the sequence of successive times spent in each site between two jumps.
%\begin{}[ht]
%\begin{center}
%\includegraphics[width=10cm,height=5cm]{expletrajY} \label{figtrajY}
%\caption{A typical trajectory of the process $Y$, $M=500$}
%\end{center}
%\end{figure}

\noindent \\

%\begin{figure}[ht]
%\begin{center}
%\input{DNA2.pstex_t} 
%\caption{ Discrete time case.} \label{fig2}
%\end{center}
%\end{figure}

\begin{figure}[ht]
\begin{center}
\includegraphics[width=10cm,height=5cm]{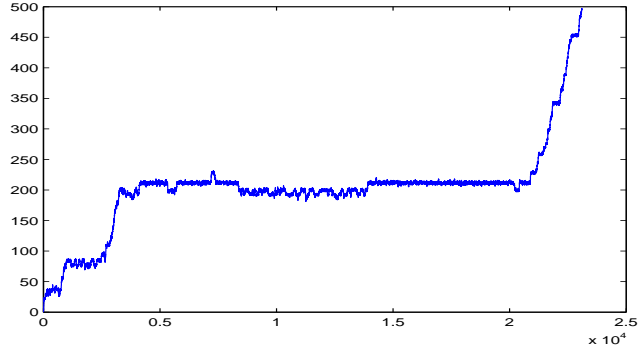}
\caption{A typical trajectory of the number of unzipping pairs, $M=500$.}
\end{center}
 \label{fig3bb}
\end{figure}

Moreover, all along the paper, we assume that  $g_0$ is injective on the first and second variables : i.e. for all $a \in \{A,T,C,G\}$ the functions
\begin{align} 
g_0(a,.):\gamma \rightarrow g_0(a,\gamma)\ \text{ and }\ g_0(.,a): \gamma \rightarrow g_0(\gamma,a)
\label{Hyp1}
\end{align}
are injective.
Note that this hypothesis matches with the experimental values of the energy (see Figure \ref{tab1}).

We describe now  briefly some results obtained by the physicists in the continuous time case.

\subsection{Some results obtained by the physicists}

In their papers \cite{Monasson0}, \cite{Monasson1} and \cite{Monasson2}, they assume first that there is no stacking effect, considering that $g_0$ at site $x$ is only a function of $b_x$ and that $(g_0(b_x),x)$ is a sequence of independent and identically distributed (i.i.d.) random variables. In this case they compute the maximum likelihood estimator for $b_x$. For a better accuracy they consider several total unzipping instead of a single one, that is to say they look at a sequence  of $R$ independent trajectories $(Y^{(l)}_.,l \leq R)$. 
In a second step they study the decreasing of the probability that this estimator gives a base sequence, and they show that this probability decreases exponentially; for all $x \leq M$, $$\Pp(b_x\neq \hat{b}_x) \leq \exp(-R/R_c(x)).$$ The constant $R_c$ is estimated numerically. For the general case (with stacking effects) they use Viterbi algorithm \cite{Viterbi} to compute the maximum of likelihood. Then they estimate the probability to be wrong with this estimator by using both analytic and numerical methods, they get a similar result than for the independent case.  \\
After some discussions with S. Cocco and R. Monasson some questions rise: is it possible to get a general and rigorous method which can be applied to all these cases ? how the choice of the force can be used in order to improve the results ? and what is the difference between the discrete and continuous time model ? We study all these questions in the present paper.

\subsection{A mathematical point of view}

First we would like to recall some basic facts for the discrete time model. If we forget, for the moment, that the state space is finite, $(X_k,\ k \in \N)$ is a random walk on a random environment on $\Z$ as Solomon defined it in \cite{Solomon}. We know, for example, that for i.i.d. sequence $(g_0(b_x),x)$, if $g_0(b_1)$ has mean zero and $g_1=0$, then $X$ is almost surely recurrent, it is transient on the other case. For the recurrent case% with the ellipticity condition
, $X$ is a Sinai's walk \cite{Sinai}, for the transient one, the first study is due to H. Kesten, M.V. Kozlov, F. Spizer \cite{KesKozSpi}. Here we are interested on what a trajectory of the walk can say about the environment, this aspect has not been studied a lot, there is a paper of O. Adelman, N. Enriquez \cite{AdeEnr} and for the special case of Sinai's walk a paper of P. Andreoletti \cite{Pierre7}. More precisely \cite{Pierre7}  shows that $g(x)$ can be estimated from a single trajectory of the walk by studying the asymptotics (in time) of the local time at site $x$,   which is the amount of time the walk spends at this site. However this approach can not be used to give informations on a particular site, typically on $g_0(b_x,b_{x+1})$ for a given $x$. \\
To move from Solomon walks to the problem asked by the physicists we have to make a sacrifice, more especially we are no longer interested in asymptotics in time. Indeed if the time goes to infinity that means that either we have to wait a very long time to reach the end of the molecule, or once it is totally unzipped it can move back to the beginning. This last case is not possible because when the end of the molecule is reached then the two separate strands can not reform the molecule properly.  In compensation, we only have to study the processes $X$ or $Y$ until they reach $M$, that is until time
\begin{align*}
\tau_M=\inf\{k>0, X_k=M\}. 
\end{align*}
 So we are interested in the discrete time process $(X_k,\ k \leq \tau_M)$ and the continuous one  $Y=(X_k,T_k,\ k \leq \tau_M)$. Note also that $M$ is the length of the DNA molecule, in term of the number of pairs, which can be big but finite. The other good news is the fact that the DNA molecule can be unzipped a large number of times, we have called, this number $R$, and we will be looking at asymptotics in this variable. Finally we are looking at $R$ independent trajectories denoted $(Z^{(l)}_{t_l},\ 1\leq l \leq R,\ 0 \leq t_l \leq \tau_M^{(l)})$ of random walks on a same unknown environment $b$ with $\tau_M^{(l)}$ the first time the walk $l$ hits $M$ ($Z$ is either $X$ or $Y=(X,T)$). 
 Also we will see that even if we assume that the $g_0$ are random, its distribution will not play an important role in our setting, essentially for two reasons the first one is the fact that the state space is finite and the second one is that we are looking at asymptotics in $R$.
% We state most of our results without any assumptions on the distribution of the sequence $b$. However to simplify the expressions we assume sometimes  that all molecules are equiprobable. \\
 %We only consider hypothesis induced by the property of a DNA molecule, first one the $g_0(b_{i},b_{i+1})$ can only takes $K$ different values, typically, by looking of Table \ref{tab1}, $K=10$. We also use the fact that there exists a compatibility between two successive energy $g_0(i,i+1)$ and $g_0(i+1,i+2)$ for a given $i$ according to Table \ref{tab1} . For example if $g_0(i,i+1)=1.06$, then $g_0(i+1,i+2)$ can only takes the values 1.78, 1.55, 2.52 and 2.22.
The method is based on the fact that,  given the trajectory of a random walk (or $R$ random walks) on an environment $b$, the probability that a given estimator $\hat{b}$ gives a good sequence (typically $\Pp(b= \hat{b})$) depends only on elementary functions of the trajectory of this random walk.
 \noindent \\ For the \emph{discrete time model}, the important quantities are the number of times $X$ goes from $x$ to $x+1$ or to $x-1$, with $x \in [1,M-1]$:
\begin{align*}
L^{+,(l)}_x&:=\sum_{k=0}^{\tau_M^{(l)}-1}\un_{X^{(l)}_k=x;X^{(l)}_{k+1}=x+1},\quad L^{-,(l)}_x:=\sum_{k=0}^{\tau^{(l)}_M-1}\un_{X^{(l)}_k=x;X^{(l)}_{k+1}=x-1},\\
L^{+,R}_x&:=\sum_{l=1}^{R}L^{+,(l)}_x\ \text{ and }\ L^{-,R}_x:=\sum_{l=1}^{R}L^{-,(l)}_x.
\end{align*}
For the  \emph{continuous time model}, we have also to consider the total time spent in each site until the instant $\tau^{(l)}_M$ (which is as for the discrete case the hitting time of $M$ for the processes $X^{(l)}$): for any $x\in[1,M-1]$,
$$\tt^{(l)}_x=\sum_{i=0}^{\tau_M^{(l)}}T^{(l)}_{i}\un_{X^{(l)}_{i}=x}\ \text{ and }\ \tt^{R}_x=\sum_{l=0}^{R}\tt^{(l)}_x.$$

We will denote by $X^R$ ($Y^R$ in the continuous case) the $\sigma$-field generated by the trajectories of the $R$ independent random walks killed when they hit the coordinate $M$. $\Pp$ denotes the probability distribution of the whole system, whereas $\Pp^{\alpha}$ is the probability distribution of the walk for a given sequence of nucleotides $\alpha$. Also $\Ep^{\alpha}$ (resp. $\Vp^{\alpha}$ for the variance) is the expectation associated to $\Pp^{\alpha}$. \\
%, whereas $P^{\alpha}$ is the probability distribution for a given sequence of nucleotides $\alpha$. 
 %Also $E^{b}$ (resp. $Var^{b}$ for the variance) is the expectation associated to $\Pp^{\alpha}$. \\
In Section \ref{sec2}, we start by the estimation base by base, we define the \emph{information} at site $x$ for both cases and show that the expression of the probability to get a given base at a site $x$ conditionally on the trajectories are a simple function of the information. Then we study the asymptotic (in $R$) of the probability that the maximum likelihood estimator gives a wrong base, we define and study a typical number of unzipping $R_c$ which measures the quality of our prediction. In a second time we are interested in the estimation of the whole molecule, we start with a general expression of the probability to get a specific sequence given the trajectories of $R$ random walks. We show that the global maximum likelihood estimator converges. Then we study the probability to make at least one mistake and then $h$ separate mistakes by considering this estimator. We focus on the continuous case, and just quote the differences with the discrete case. \\
In Section \ref{sec3}, we study some possible improvements. The first one consists on a local modification of the force in order to trap the system in a specific region. It has a direct effect on the time spent in this region and therefore on the quality of the prediction. For the second one we also modify the force, it is now function of the binding energies, and also of the space. It allows a fast unzipping till the bases we are interested, and a fast decreasing of the probability to be wrong.

%%%%%%%%%%%%%%%%%%%%%%%     DEBUT MODIF    %%%%%%%%%%%%%%%%%%%%%%%%%%%%%%%%%%%%%%%%%%%%%%%%%%%%%%%%

\section{Bayes estimator, asymptotics in $R$ and typical number of needed unzipping $R_c$ \label{sec2}}

Most of the results of this section are based on the fact that we can compute easily the joint distribution $(L^{+,(1)}_x, L^{-,(1)}_x=L^{+,(1)}_{x-1}-1)$,  in fact it is not more difficult to get the joint distribution $L^+:=(L^{+,(1)}_x, 1\leq x \leq M-2)$ and as we have not found it in the literature, we first prove the following lemma for one random walk :
\begin{lemme}
\label{loiLp}
If we denote $k=(k_x\ ,\ x\in\{1,\dots,M-1\})$ with $k_{M-1}=1$, then %the distribution of $L^+:=L^{+,(1)}$ is
\begin{align*}
 \Pp^b\left(L^{+}=k\right)%=&p_{M-1}(1-p_{M-1})^{k_{M-2}-1}\prod_{x=2}^{M-2}\binom{k_x+k_{x-1}-2}{k_x-1}p_x^{k_x}(1-p_x)^{k_{x-1}-1},\\
=&\prod_{x=2}^{M-1}\binom{k_x+k_{x-1}-2}{k_x-1}p_x^{k_x}(1-p_x)^{k_{x-1}-1}.
\end{align*}
 In particular, for $x\in\{2,\cdots,M-1\}$, 
\begin{align}
 & \Pp^b\left(L^{+}_x=k_x,L^{-}_x=k_{x-1}\right)= \Pp^b\left(L^{+}_x=k_x,L^{+}_{x-1}=k_{x-1}+1\right)\nonumber\\ 
 =&\binom{k_x+k_{x-1}-1}{k_x-1} (1-p_x)^{k_{x-1}} (p_x (1-\bar{p}_x))^{k_x-1}(p_x\bar{p}_x)\label{patapo}
 \end{align}
 where for simplicity we denote $L^{+}_x :=L^{+,(1)}_x, L^{-}_x :=L^{-,(1)}_x$ and
 \begin{align}
 \frac{1}{\bar{p}_{x}} & =\frac{1}{\bar{p}_{x}(b,f)}:=\frac{1}{P^b_{x+1}(\tau_{x}>\tau_M)}=\sum_{k=x+1}^{M-1}\exp\left(\beta\left(g(k)-g(x)\right)\right)+1, \label{pbar}
 \end{align}
with $\tau_x:=\inf\{k>0, X_k=x\}$. It is then easy to compute the following means and variances
\begin{align}
  \Ep^b\left(L_x^+\right) & = \frac{1}{\bar{p}_x},\ E^b\left(L_x^-\right)= \frac{e^{\beta \dg(b_x,b_{x+1})}}{\bar{p}_x}, \Ep^b\left(S^{(1)}_x\right)=\frac{e^{\beta g_0(b_x,b_{x+1})}}{r\bar{p}_x}, \nonumber \\
  \Vp^b\left(L_x^+\right)& =\frac{1}{\bar{p}_x}\left(\frac{1}{\bar{p}_x}-1\right)\text{ and }  \Vp^b\left(S^{(1)}_x\right)=\frac{e^{2\beta g_0(b_x,b_{x+1})}p_x}{r^2\bar{p}_x}. 
 \end{align} 
\end{lemme}
 \begin{proof} 
 Formula \eqref{patapo} of the lemma can easily be obtained by using the Markov property of $X$ for a given sequence $b$, the mean and the variance of $L^+_x$ and $S^{(1)}_x$ are direct consequences. Therefore we just prove the expression of the joint distribution of $L^+$. Define for $n\geq1$, the event
%Let $k_{M-1}=1$, and $k_{-1}=0$, it means that there is only one jump from $M-1$ to $M$ and that the walk is reflected at $0$. 
$$A_n:=\bigcap_{x=n}^{M-1}\{L^+_x=k_{x}\}$$ where $k_{M-1}=1$ (there is always only one jump from $M-1$ to $M$). Then,
\begin{align*}
 \Pp^b\left(L^{+}=k\right)&=\Pp^b\left(A_{1}\right)=\Pp^b\left(L^+_{1}=k_{1}\big|A_{2}\right)\Pp^b\left(A_{2}\right)\\
 &=\Pp^b\left(L^+_{1}=k_{1}\big|L^+_{2}=k_{2}\right)\Pp^b\left(A_{2}\right)
\end{align*}
where the second equality comes from the Markov property of the walk $X$ given $b$. Formula \eqref{patapo} implies for any $x\in\{2,\cdots,M-1\}$,
\begin{align*}
 \Pp^b\left(L^{+}_{x-1}=k_{x-1}\big|L^{+}_x=k_x\right) =\binom{k_x+k_{x-1}-2}{k_x-1} (1-p_x)^{k_{x-1}-1}p_x^{k_x},
 \end{align*}
thus,
\begin{align*}
 \Pp^b\left(L^{+}=k\right)&=\binom{k_2+k_{1}-2}{k_2-1} p_2^{k_2}(1-p_2)^{k_{1}-1}\Pp^b\left(A_{2}\right)
\end{align*}
and we get the result of Lemma \ref{loiLp} recursively.
\end{proof}

\subsection{Prediction site by site}

\emph{In this section we always assume that $f$ is constant.} 
Let us begin with a general proposition true for the continuous and the discrete time cases, then we discuss the differences between the two cases.
First we define the following function $i_x$, called \emph{local information at site $x$}  of the system, it differs for the two cases.
Let $x\in\{2,\cdots,M-1\}$ and $(\alpha_{x-1},\alpha_x,\alpha_{x+1}) \in \{A,T,C,G\}^3$.\\
\emph{For the discrete case,} the information is defined by
\begin{align*}
&i_x(\alpha_{x-1},\alpha_x,\alpha_{x+1}):=\\
L_x^{+,R}&\log (1+e^{\beta \dg(\alpha_x ,\alpha_{x+1}})+L_x^{-,R}\log  (1+e^{-\beta \dg(\alpha_x ,\alpha_{x+1})}) \\
+L_{x-1}^{+,R}&\log (1+e^{\beta \dg(\alpha_{x-1},\alpha_x)})+L_{x-1}^{-,R}\log(1+e^{-\beta \dg(\alpha_{x-1},\alpha_x)}).
\end{align*}
and \emph{for the continuous case,} by
\begin{align*}
i_x(\alpha_{x-1},\alpha_x,\alpha_{x+1})&:=
\beta g_0 (\alpha_x,\alpha_{x+1}) L^{+,R}_x+S_x^R r e^{-\beta g_0(\alpha_x,\alpha_{x+1})} \\
&+\beta g_0 (\alpha_{x-1},\alpha_x) L^{+,R}_{x-1}+S^R_{x
-1} re^{-\beta g_0(\alpha_{x-1},\alpha_x)}. 
\end{align*}
We are now ready to state the

\begin{proposition}
\label{quipu1} For all $x \in \{2, \cdots,M-1\}$, and for $\alpha_x \in \{A,T,C,G\}$, denoting $b^x=(b_1,b_2,\cdots,b_{x-1},b_{x,+1},\cdots,b_{M-1})$, we have
\begin{align}
\Pp\left(b_x= \alpha_x |Z^R,\ b^x \right) 
&=  \frac{\exp(-I_x(\alpha_x,b))}{\sum_{\bar{\alpha}_x}\exp(-I_x(\bar{\alpha}_x,b))}  \label{ptbw}
\end{align}
where
\begin{align*}
I_x(u,b)=I_x(u,b)(Z^R)&:=i_x(b_{x-1},u,b_{x+1})-\log \Pp(b_x=u|b^{x}),
\end{align*}
and $Z^R$ is either $X^R$ for the discrete case or $Y^R$ for the continuous one.
The maximum likelihood estimator $\hat{b}_x$ for $b_x$, is given by:
\begin{align}
%\hat{b}(x)=S \un_{\{J_x(S,b)>J_x(W,b) \}}+W \un_{\{ J_x(S,b) \leq J_x(W,b) \}}. \label{max1}
\hat{b}_x=\sum_{\alpha_x \in \{A,T,C,G\}} \alpha_x \un_{\{I_x(\alpha_x,b)=\min_{\bar{\alpha}_x}I_x(\bar{\alpha}_x,b) \}}. \label{max1}
\end{align}
Assume that Hypothesis (\ref{Hyp1}) is satisfied then the maximum likelihood estimator converges almost surely to $b_x$.
Moreover,
\label{quipu2} %For all $x$, all $g_0^x$ and if $\hat{g}_0(x)$ is given by (\ref{max1}) then 
\begin{align}
\lim_{R\rightarrow\infty}-\frac{1}{R}\log\Pp\left(\hat{b}_x \neq {b}_x|Z^R,b^x \right)&= 1/R_c(x)>0; \label{eq10b}
\end{align} 
$R_c(x)$ is called \emph{the typical number of random walks at site $x$}. %, it is defined by: \\ 
%\begin{align} 
%\frac{1}{R_c(x)}:= \lim_{R \rightarrow +�\infty } \frac{- \log \Pp\left(\hat{b}_x \neq {b}_x|Z^R,b^x \right)}{R}.
%\end{align}
For the discrete case,  $\Pp$-almost surely,
\begin{align*}
 \frac{1}{R_c(x)} = \frac{ \Delta {G}^-(b_{x-1})}{\bar{p}_{x-1}}+\frac{ \Delta {G}^+(b_{x+1})}{\bar{p}_x},
%& \Delta {G}^- \left( \frac{1}{\bar{p}_x}+ \frac{1}{\bar{p}_{x-1}}\right) \leq  1/R_c(x) \leq  \Delta {G}^+  \left( \frac{1}{\bar{p}_x}+ \frac{1}{\bar{p}_{x-1}}\right)
\end{align*}
and for $R$ large enough
\begin{align*}
 \frac{1}{R_c(x)} = \Delta {G}^-(b_{x-1}) \frac{L_{x-1}^{+,R}}{R}+\Delta {G}^+(b_{x+1})\frac{L_x^{+,R}}{R}+\epsilon_x(R), 
%& \Delta {G}^- \left( \frac{1}{\bar{p}_x}+ \frac{1}{\bar{p}_{x-1}}\right) \leq  1/R_c(x) \leq  \Delta {G}^+  \left( \frac{1}{\bar{p}_x}+ \frac{1}{\bar{p}_{x-1}}\right)
\end{align*}
where $ \Delta {G}^+(b_{x+1})$ and $ \Delta {G}^+(b_{x-1})$ are two positive numbers (see (\ref{dG})).  
For the continuous case we get the same expression but replacing the constant $ \Delta {G}^+(b_{x+1})$ (respectively $ \Delta {G}^+(b_{x-1})$) by $ \Delta {F}^+(b_{x+1})$ (respectively $ \Delta {F}^+(b_{x-1})$), see also their expression in (\ref{dF}). Also for the discrete and continuous case we have $\epsilon_x(R)\approx ( R \log \log R)^{1/2} {L^{+,(R)}_{x}}/{R} $%or $\epsilon_x(R)\approx ( R \log \log R)^{1/2}/{\bar{p}_x} $
. We denote $a(R) \approx d(R)$ if there exists a positive bounded  number $c$ such that $a(R) = c*d(R)$. 
\end{proposition}

%( R \log \log R)^{1/2} ( (Var^{b} L^{+,(1)}_x)^{1/2}+(Var^{b} S^{(1)}_x)^{1/2})

We first prove the result and then discuss about the expression of $R_c(x)$.
\begin{proof}
 %We will give the proof of a more general result for the continuous case in Section \ref{sec2.2},  however 
We only give a proof in the discrete case. 
 Formula (\ref{ptbw}) is a simple consequence of Bayes formula together with Lemma \ref{loiLp} and the expression of $\hat{b}_x$ follows. By the strong law of large number (LLN), $\Pp^{b}$-almost surely, %\textbf{[il y a un probl�me ici, en effet la loi des grands nombre sous $\Pp[|Z^R,\ b^x]$ moyenne sur tous les $g_0$ y compris $g_0(b_x,.)$, en fait il faut supposer les $b_x$ inconnues mais pas al\'eatoire !]}
  \begin{align*}
& \lim_{R \rightarrow + \infty}\frac{1}{R} I_x({\alpha}_x,b)  \nonumber\\
&= \frac{1}{\bar{p}_{x}}\left(\log(1+e^{\beta \dg(\alpha_{x} ,b_{x+1})})+e^{\beta \dg(b_{x},{b}_{x+1})}\log  (1+e^{-\beta \dg(\alpha_x ,b_{x+1})})\right)+ \nonumber \\
&  \frac{1}{\bar{p}_{x-1}} \left(\log(1+e^{\beta \dg( b_{x-1},\alpha_x)})+e^{\beta \dg(b_{x-1},{b}_x)}\log  (1+e^{-\beta \dg(b_{x-1},\alpha_x )})\right).
\end{align*} 
 Recall that $\dg$ is defined in (\ref{Delta}).
 %applying again the law of large number we turn back to an expression involving $L_{x-1}^{+,R}$ and $L_{x}^{+,R}$:
%\begin{align}
%& \lim_{R \rightarrow + \infty}\frac{1}{R}I_x({\alpha}_x,b)  \nonumber \\
%&=  \left(\log(1+e^{\beta \dg(\alpha_{x} ,b_{x+1})})+e^{\dg(b_{x},{b}_{x+1})}\log  (1+e^{-\beta \dg(\alpha_x ,b_{x+1})})\right)  \lim_{R \rightarrow + \infty}\frac{1}{R}L_{x}^{+,R} \nonumber \\
%& \left(\log(1+e^{\beta \dg( b_{x-1},\alpha_x)})+e^{\dg(b_{x-1},{b}_x)}\log  (1+e^{-\beta \dg(b_{x-1},\alpha_x )})\right)   \lim_{R \rightarrow + \infty}\frac{1}{R}L_{x-1}^{+,R}
%\end{align} 
%where $L_{x-1}^{+}$
As the function $$x\rightarrow\log(1+x)+c\log(1+1/x)$$ is minimal iff $x=c$, asymptotically the right-hand side of the previous equality is minimal iff $\dg(b_{x-1},\alpha_x )=\dg(b_{x-1},b_x)$ and $\dg(\alpha_{x},b_{x+1})=\dg(b_{x},b_{x+1})$, therefore with Hypothesis (\ref{Hyp1}), $\alpha_x=b_x$ and the estimator $\hat{b}_x$ is almost surely convergent.

Now we are interested in the difference $ \lim_{R \rightarrow + \infty}\frac{1}{R} (I_x(\hat{b}_x,b)-I_x(\bar{\alpha}_x,b) )$, first let us define the function %$G_a:\R \rightarrow \R_+$, 
$$ G_a(u):=\log\left(\frac{1+e^{\beta u}}{1+e^{\beta a}}\right)+e^{\beta  a}\log\left(\frac{1+e^{-\beta u}}{1+e^{-\beta a}}\right),$$ 
notice that $G_a(x)$ is positive for all $x\neq a$ and $G_a(a)=G'_a(a)=0$.
%so we can write, 
% \begin{align}
%&I_x({\alpha}_x,b)-I_x(\bar{\alpha}_x,b)  \nonumber \\
%&= L_{x}^{+,R} G_{\dg(b_{x},{b}_{x+1})} \left(\dg(\alpha_{x},b_{x+1})\right)+ L_{x-1}^{+,R} G_{\dg(b_{x-1},{b}_{x})} \left(\dg(b_{x-1},\alpha_{x})\right). 
%\end{align} 
%Assume for the moment that  $\hat{b}_x$ converge almost surely when $R$ tends to infinity to $b_x$, this fact will  be proved in a more general case in Theorem \ref{estiglob};  we get that 
$\Pp^{b}$ almost surely  for all $\bar{\alpha}_x \neq \hat{b}_x$
 \begin{align*}
&\lim_{R \rightarrow + \infty}\frac{1}{R}(I_x(\hat{b}_x,b)-I_x(\bar{\alpha}_x,b))  \nonumber \\
&=  -\frac{1}{\bar{p}_{x}}G_{\dg(b_{x},{b}_{x+1})} \left(\dg(\bar{\alpha}_{x},b_{x+1})\right)- \frac{1}{\bar{p}_{x-1}} G_{\dg(b_{x-1},{b}_{x})} \left(\dg(b_{x-1},\bar{\alpha}_{x})\right), 
\end{align*} 
By Hypothesis (\ref{Hyp1}) of local injectivity of $g_0$, we get that $\Pp^{b}$-almost surely, for all $\bar{\alpha}_{x} \neq \hat{b}_{x}$, $\lim_{R \rightarrow + \infty}\frac{1}{R}(I_x(\hat{b}_x,b)-I_x(\bar{\alpha}_x,b))$ is  strictly negative.
Also notice that 
\begin{align*} 
\sum_{\bar{\alpha}_x \neq \alpha_x}e^{I_x(\alpha_x,b)-I_x(\bar{\alpha}_x,b)} &= e^{I_x(\alpha_x,b)-I_x(\hat{b}_x,b)}\sum_{\bar{\alpha}_x \neq \alpha_x}e^{I_x(\hat{b}_x,b)-I_x(\bar{\alpha}_x,b)} \\
&= e^{I_x(\alpha_x,b)-I_x(\hat{b}_x,b)}\left(1+\sum_{\bar{\alpha}_x \neq \alpha_x, \hat{b}_x}e^{I_x(\hat{b}_x,b)-I_x(\bar{\alpha}_x,b)}\right),
\end{align*}
therefore we have that $\Pp^{b}$-almost surely for $R$ large enough 
\begin{align} 
& e^{\frac{R}{\bar{p}_{x-1}}G_{\dg(b_{x},{b}_{x+1})} \left(\dg(\alpha_{x},b_{x+1})\right) + \frac{R}{\bar{p}_{x-1}} G_{\dg(b_{x-1},{b}_{x})} \left(\dg(b_{x-1},\alpha_{x})\right)   -e_x(R)} \nonumber \\  & \leq \sum_{\bar{\alpha}_x \neq \alpha_x}\exp(I_x(\alpha_x,b)-I_x(\bar{\alpha}_x,b)) \nonumber \\
& \leq 4  e^{\frac{R}{\bar{p}_{x-1}}G_{\dg(b_{x},{b}_{x+1})} \left(\dg(\alpha_{x},b_{x+1})\right) + \frac{R}{\bar{p}_{x-1}} G_{\dg(b_{x-1},{b}_{x})} \left(\dg(b_{x-1},\alpha_{x})\right)   +e_x(R)},  \nonumber
  \end{align}
where $e_x(R)$  is the error we make by using the LLN and by the presence of the $\log \Pp(b_x=u|b^{x})$ in the expression of $I$, we examine this term at the end of the proof. Now define
\begin{align} 
 \Delta {G}^+(b_{x+1}) &:=\max_{\alpha_x\neq b_x}(G_{\dg(b_x,b_{x+1})}(\dg(\alpha_x,b_{x+1})), \label{dG} \\
 \Delta {G}^-(b_{x-1}) &:=\max_{\alpha_x\neq b_x}(G_{\dg(b_{x-1},b_x)}(\dg(b_{x-1},\alpha_x)), \nonumber
\end{align}
and finally notice that $ \Pp\left( b_x \neq \hat{b}_x |Z^R,\ b^x \right)$ can be written like 
\begin{align}
 & \Pp\left( b_x \neq \hat{b}_x |Z^R,\ b^x \right) \nonumber \\ & = \sum_{{\alpha}_x \neq \hat{b}_x } \Pp\left( b_x =\alpha_x |Z^R,\ b^x \right) \nonumber \\
&=  \sum_{{\alpha}_x  \neq \hat{b}_x } \left(1+ \sum_{\bar{\alpha}_x \neq \alpha_x}\exp(I_x({\alpha}_x,b)-I_x(\bar{\alpha}_x,b)) \right)^{-1},
\end{align}
we get that $\Pp^b$ almost surely for $R$ large enough
\begin{align*}
&  - \log \left(1+4e^{\frac{R}{\bar{p}_{x-1}} \Delta {G}^-(b_{x-1}) +\frac{R}{\bar{p}_{x}} \Delta {G}^+(b_{x+1}) + e_x(R) } \right) \leq \log  \Pp\left( b_x \neq \hat{b}_x |Z^R,\ b^x \right)  \nonumber \\ & \leq \log 4- \log \left(1+e^{\frac{R}{\bar{p}_{x-1}}\Delta {G}^-(b_{x-1})+\frac{R}{\bar{p}_{x}}\Delta {G}^+(b_{x+1})  + e_x(R)} \right),
 \end{align*}
that can be written like:
 \begin{align*}
  \left |\frac{-\log  \Pp\left( b_x \neq \hat{b}_x |Z^R,\ b^x \right)}{R}+ \frac{1}{R_c(x)} \right|
 \leq \frac{e_x(R)+\textrm{const}}{R},
 \end{align*}
%\begin{align} 
% \Delta {G}^-(b_{x-1}) &:=\min_{\bar{\alpha}_x\neq {\alpha}_x}(G_{\dg(b_{x-1},\alpha_x)}(\dg(b_{x-1},\bar{\alpha}_x)),\nonumber\\
%  \Delta{G}^+(b_{x+1}) &:=\min_{\bar{\alpha}_x\neq {\alpha}_x}(G_{\dg(\alpha_x,b_{x+1})}(\dg(\bar{\alpha}_x,b_{x+1})).\nonumber
%\end{align}
where "const" is a constant real number.

To finish the proof we have to study $e_x(R)$. By the iterated logarithm law (ILL) $e_x(R)\approx (R \log \log R)^{1/2} ( (\Vp^{b} L^{+,(1)}_x)^{1/2}+(\Vp^{b} L^{-,(1)}_x)^{1/2}) $  $-\log \Pp(b_x=u|b^{x}) $ where $(\Vp^{b}(L_x^+))^{1/2}$ as well as  $(\Vp^{b}(L_x^-))^{1/2}$ behaves like $1/\bar{p}_x$ (see Lemma \ref{loiLp}). Then  $\epsilon_x(R)  \approx \sqrt{R\log\log R}/\bar{p}_x$, this gives the expression of the proposition by using the LLN. Also to move from the result under the measure $\Pp^b$  to the result under $\Pp$ we just notice that what we get is true for all sequences $b$.\\
%by way of the law of large number again we finalyt get the main contribution $\frac{R}{R_c(x)}$ on the exponential of \ref{eq10b}. by way of the law of large number again.
%\textbf{[ Il faut parler du $\epsilon$ et donc de la LLI]}
\end{proof}

Notice that $1/R_c(x)$ is the rate function in the large deviation theory so the above proposition gives  informations on the decrease of the probability to be wrong%, we now separate the two cases, and discuss about it
.\\
\emph{The discrete case}. We have obtained that $\Pp$-almost surely for $R$ large enough
\begin{align*} \frac{1}{R_c(x)} =
 \frac{1}{\bar{p}_{x-1}}\Delta {G}^-(b_{x-1}) +\frac{1}{\bar{p}_{x}} \Delta{G}^+(b_{x+1}),
 \end{align*}
 %$$ 1/R_c(x) \geq  \beta \left( K^+ \frac{1}{\bar{p}_x} \wedge \frac{1}{\bar{p}_{x-1}} + K^- \frac{1}{\bar{p}_{x-1}}\wedge \frac{1}{\bar{p}_{x-2}} \right), $$ 
first note that we want $\Delta {G}^-(b_{x-1})$ and $\Delta {G}^+(b_{x+1})$ as large as possible, unfortunately they may be very small, indeed
\begin{align*} 
G_a(u)={\beta^2 (u-a)^2} \frac{ G_a''(a)}{\beta^2}+o(u-a)^2,
\end{align*} 
thus  when the correct energy $g_0(b_{x-1},b_x)$ (take for example 1,78 in the table of energies Figure \ref{tab1}) is close to another one (take 1,55) we have
\begin{align*}
\Delta {G}^-(b_{x-1}) &\thickapprox \beta^2 \min_{\alpha_x\neq b_x}(\dg(b_{x-1},\alpha_x)-\dg(b_{x-1},b_x))^2 \\
&=  \beta^2 \min_{\alpha_x\neq b_x}(g_0(b_{x-1},\alpha_x)-g_0(b_{x-1},b_x))^2, 
\end{align*}
so $\Delta {G}^-(b_{x-1})$ and $\Delta {G}^-(b_{x-1})$ can be small. However, this is not the only and worst case. Indeed, assume $u<0$ and $a<0$, then for large $\beta$,
\begin{align} 
G_a(u) \thickapprox \exp (\beta a)(\exp(\beta(u-a))-1-\beta(u-a)), \label{M10}
\end{align} 
which exponentially decreases with $\beta$. This situation may appear when $f$ is large and the binding energy at site $x$ of the molecule is weak. We will see in Section \ref{sec3} a method to avoid this situation. \\
Turning back to the expression of $1/R_c(x)$, we also notice that
\begin{align}
\frac{1}{\bar{p}_x}&\geq \exp(\beta(\max_{ x+1 \leq l\leq M-1}(g(l)-g(x)) )) \nonumber \\
&= \exp\left(\beta M_x  \right), \label{Mx} 
\end{align}
with $M_x := \max_{ x+1 \leq l\leq M-1}\left\{\sum_{k=x+1}^lg_0(b_k,b_{k+1})-(l-x)g_1(f) \right\}$.
So, as expected, the convergence is better if there are obstacles in the path from $x$ to $M$. Finally, we have $\Pp$-almost surely
\begin{align} 
1/R_c(x)   \geq 
 \exp\left(\beta M_{x-1} \right)\Delta {G}^- +\exp\left(\beta M_{x} \right) \Delta{G}^+,
\end{align}
with $$ \Delta {G}^-:=\min\{\Delta {G}^-(\gamma), \gamma \in \{A,T,C,G\} \}$$ and $$ \Delta {G}^+:=\min\{\Delta {G}^+(\gamma), \gamma \in \{A,T,C,G\} \}.$$

\noindent \\ \emph{ Formula useful for the estimation.} As we have seen above, $R_c(x)$ characterizes locally the environment. % However, this information is available only a posteriori.
However, what is really important to control the quality of estimation at a point $x$ is not the number of walks $R$ but the total number of passages at this point, $L^R_x:=L^+_x+L^-_x$. So  we define the \emph{typical number of visits at site $x$}, $L_c(x)$ by  
\begin{align} 1/L_c(x):=\lim_{R\rightarrow +\infty } \frac{- \log \Pp\left(\hat{b}_x \neq {b}_x|Z^R,b^x \right)}{L^R_x},
\end{align}
and we get $\Pp$-almost surely
$$\frac{1}{L_c(x)}\geq \frac{1}{2}( \Delta G^+\wedge \Delta G^-).$$

\noindent \\ \emph{ Total amount of time to reach $M$.} An other important factor is the time required to unzip totally $R$ times the DNA molecule. It should not be too large. This time is given by:
\begin{align}
\tau_M^{R}=\sum_{l=1}^R{\tau_M^{(l)}}=\sum_{l=1}^R \sum_{x=1}^{M-2}(L_{x-1}^{+,(l)}+L_x^{+,(l)}-1). \label{time2r}
\end{align}
And by the LLN, $\Pp$ almost surely
\begin{align*}
R \exp(\beta \max _x M_x) \lesssim E^b\left[\tau_M^R\right] =&R \sum_{x=1}^{M-2}\left(\frac{1}{\bar{p}_{x-1}}+\frac{1}{\bar{p}_x}-1\right) \\
\lesssim & R M \exp(\beta \max _x M_x).
\end{align*}
So, as seen in the previous  paragraph (see (\ref{Mx})), large $\beta$ can lead to a better prediction, however it slows down the system. Of course it is worse if there is large obstacles between $x$ and  $M$  because in this case $M_x$ is large too. %The best compromise we can find here is to keep $\beta$ large enough in order to have a good decrease of the probability to be wrong associated to a force large enough in order to have $ \max _x M_x \thickapprox 1$.

\noindent \\ \emph{ The continuous time case.}
\noindent Like for the discrete case we first define a function $F:\R \rightarrow \R_+$ by
\begin{align}
F(u)&=e^{\beta u}-1-\beta u\text{ and } \nonumber \\
\Delta F^-(\gamma)&=\min\left(F(g_0(\gamma,u)-g_0(\gamma,v)), u, v\in \{A,T,C,G\}, u\neq v\right), \nonumber \\
\Delta F^+(\gamma)&=\min\left(F( g_0(u,\gamma)-g_0(v,\gamma)), u, v\in \{A,T,C,G\}, u\neq v\right), \label{dF} \\
\Delta F^-&=\min\left(\Delta F^-(\gamma),\gamma\in \{A,T,C,G\}\right), \nonumber \\
\Delta F^+&=\min\left(\Delta F^+(\gamma),\gamma\in \{A,T,C,G\}\right) \nonumber.
\end{align}
Then a study, similar to the discrete case, leads to, $\Pp$ almost surely
$$1/R_c(x)\geq \frac{\Delta F^+}{\bar{p}_x}+\frac{\Delta F^-}{\bar{p}_{x-1}}\text{ and } 1/L_c(x)\geq \frac{1}{2} (\Delta F^+\wedge \Delta F^-).$$
Note that the bad case observed for the discrete time model (see equation \eqref{M10}) does not appear here, however when $(u-a)$ is small, $F(u-a)$ is as $G_a(u)$ of the order of $(a-u)^2$.

In the next section we look at the entire molecule, we define global information and study the decreasing of the probability to make a mistake by using the global maximum likelihood estimator.

 \subsection{Inferring the whole molecule \label{sec2.2}}

Define the \emph{global information} $I$ of the whole molecule, let $\alpha \in \{A,T,C,G\}^M$, for the  \emph{discrete case $X^R$,}
\begin{align*}
 I(\alpha)&:=-\log \Pp(b=\alpha)+\\
 &\sum_{x=1}^{M-1} L_x^{+,R}\log (1+e^{\beta \dg(\alpha_x,\alpha_{x+1})})+L_x^{-,R}\log  (1+e^{-\beta \dg(\alpha_x,\alpha_{x+1})}).
\end{align*}
and for \emph{the continuous case $Y^R=(X^R,T^R)$,} 
\begin{equation}
I(\alpha)=-\log \Pp(b=\alpha)+\sum_{x=1}^{M-1}\beta g_0(\alpha_{x},\alpha_{x+1})L^{+,R}_x+re^{-\beta g_0(\alpha_x,\alpha_{x+1})}S^R_x
\end{equation}
The global maximum likelihood estimator converges to $b$ :
\begin{theoreme}\label{estiglob}
For any $\alpha\in\{A,T,C,G\}^M$ with $\alpha_1=b_1$, we have: 
\begin{equation}
    \Pp\left(b=\alpha\big|Z^R,\ b_1\right)= \frac{e^{-I(\alpha)}}{\sum_{%\bar{\alpha}\in\{A,T,C,G\}^M,\\ 
    \bar{\alpha}}e^{-I(\bar{\alpha})}}.
%   =\Pp\left(b=\alpha\big|L^{+,R},\ S^R\right) 
\end{equation}
The global maximum likelihood estimator is the element $\tilde{b}$ of $\{A,T,C,G\}^M$ which minimizes the function $I$. %If the function $$G_0:\ \alpha\rightarrow \left(g_0(\alpha_x,\alpha_{x+1}),\ x\in\{1,\cdots,M-1\}\right)$$
%is injective, \\
\\
Assume that (\ref{Hyp1}) is satisfied then the global maximum likelihood estimator converges almost surely to the DNA chain $b$.
\end{theoreme}

\begin{proof}
We only give the proof for the continuous case, the discrete one uses the same ideas.
 For a realization of $Y^R=\left(X^{(1)},\cdots,\ X^{(R)}, \right.$ $\left. \ T^{(1)},\cdots,\ T^{(R)}\right)$, Bayes Lemma gives :
\begin{align*}
 \Pp&\left(b=\alpha|Y^R=y\right)=\Pp\left(b=\alpha\bigg|\bigcap_{l=1}^R\left\{X^{(l)}=x^{(l)},\ T^{(l)}=t^{(l)}\right\}\right)\\
 &=\frac{\Pp\left(\bigcap_{l=1}^R\left\{X^{(l)}=x^{(l)},\ T^{(l)}=t^{(l)}\right\}\bigg|b=\alpha\right)\Pp(b=\alpha)}{\Pp\left(\bigcap_{l=1}^R\left\{X^{(l)}=x^{(l)},\ T^{(l)}=t^{(l)}\right\}\right)}.\\
\end{align*}
Notice that we still use $\Pp$ to denote a probability density.
When $X^{(l)}_i=x^{(l)}_i$, and the environment $\alpha$ are given, $T^{(l)}_i$ is an exponential variable, independent of the other $(X^{(k)},\ T^{(k)})$, and of parameter $r\left(e^{-\beta g_0\left(\alpha_{x_i^{(l)}},\alpha_{x_i^{(l)}+1}\right)}+e^{-\beta g_1(f)}\right)$ if $x^{(l)}_i\in\{2,\cdots,M-1\}$ or $re^{-\beta g_0(\alpha_{1},\alpha_{2})}$ if $x^{(l)}_i=1$. Thus,
\begin{align*}
&\Pp\left(\bigcap_{l=1}^R\left\{X^{(l)}=x^{(l)},\ T^{(l)}=t^{(l)}\right\}\bigg|b=\alpha\right)\\%&=P_b\left(\bigcap_{l=1}^R\left\{T^{(l)}=t^{(l)}\right\}\bigg|b=\alpha,\bigcap_{l=1}^R\left\{X^{(l)}=x^{(l)}\right\}\right)\\
=&\Pp\left(\bigcap_{l=1}^RX^{(l)}=x^{(l)}\bigg|b=\alpha\right)\prod_{l=1}^{R}\prod_{i=1}^{\tau_M^{(l)}-1}\Pp(T_i^{(l)}=t^{(l)}_i|b=\alpha,\ X^{(l)}_i=x^{(l)}_i)\\
=&\Pp\left(\bigcap_{l=1}^RX^{(l)}=x^{(l)}\bigg|b=\alpha\right)(re^{-\beta g_0(\alpha_1,\alpha_{2})})^{l^{R}_{1}}e^{-s^R_1re^{-\beta g_0(\alpha_1,\alpha_{2})}}\\
\times&\prod_{x=2}^{M-1}(r(e^{-\beta g_0(\alpha_x,\alpha_{x+1})}+e^{-\beta g_1(f)}))^{l^{R}_{x}}e^{-s^R_xr(e^{-\beta g_0(\alpha_x,\alpha_{x+1})}+e^{-\beta g_1(f)})}
\end{align*}
$$
\text{where}\ l^{R}_{i}=\sum_{l=1}^R\sum_{k=1}^{\tau_M^{(l)}-1}\un_{x^{(l)}_k=i}\ \text{ and }\ s^R_i=\sum_{l=1}^R\sum_{k=1}^{\tau_M^{(l)}-1}t^{(l)}_{x^{(l)}_k}\un_{x^{(l)}_k=i}.
$$
Moreover
\begin{align*}
&\Pp\left(\bigcap_{l=1}^RX^{(l)}=x^{(l)}\bigg|b=\alpha\right)\\
=&\prod_{x=2}^{M-1}\left(\frac{e^{-\beta g_0(\alpha_x,\alpha_{x+1})}}{e^{-\beta g_0(\alpha_x,\alpha_{x+1})}+e^{-\beta g_1(f)}}\right)^{l^{+,R}_x}\left(\frac{e^{-\beta g_1(f)}}{e^{-\beta g_0(\alpha_x,\alpha_{x+1})}+e^{-\beta g_1(f)}}\right)^{l^{+,R}_{x-1}-1}\\
 =&\prod_{x=2}^{M-1}\frac{e^{-l^{+,R}_x\beta g_0(\alpha_x,\alpha_{x+1})-(l^{+,R}_{x-1}-1)\beta g_1(f)}}{\left(e^{-\beta g_0(\alpha_x,\alpha_{x+1})}+e^{-\beta g_1(f)}\right)^{l^R_x}}
\end{align*}
where
$$
l^{+,R}_{i}=\sum_{l=1}^R\sum_{k=1}^{\tau_M^{(l)}}\un_{x^{(l)}_k=i,\ x^{(l)}_{k+1}=i+1}.$$%\ \text{ and }\ l^{R}_{i}=l^{+,R}_{i}+l^{+,R}_{i-1}-1.$$ 
Then we have the following equality
\begin{align*}
&\Pp\left(\bigcap_{l=1}^R\left\{X^{(l)}=x^{(l)},\ T^{(l)}=t^{(l)}\right\}\bigg|b=\alpha\right)\\
=&\prod_{x=2}^{M-1}r^{l^{R}_{x}}e^{-s^R_xr(e^{-\beta g_0(\alpha_x,\alpha_{x+1})}+e^{-\beta g_1(f)})-l^{+,R}_x\beta g_0(\alpha_x,\alpha_{x+1})-(l^{+,R}_{x-1}-1)\beta g_1(f)} \\
&\times r^{l^R_1}e^{-s^R_1re^{-\beta g_0(\alpha_1,\alpha_{2})}-l^{+,R}_1\beta g_0(\alpha_1,\alpha_{2})}.
\end{align*}
Assembling the different expressions we get the formula for $\Pp(b=\alpha|Y^R)$.

We now prove the convergence of the maximum likelihood estimator. According to Lemma \ref{loiLp}, the LLN and the LIL, for any $x\in\{1,\cdots,M-1\}$, $\Pp^b$ almost surely for $R$ large enough
\begin{align*}
L^{+,R}_x=\frac{R}{\bar{p}_x(b)}+\epsilon_x(R)\ \text{ and}\ S^R_x&=\frac{Re^{\beta g_0(b_x,b_{x+1})}}{r\bar{p}_x(b)}+\epsilon_x(R). %\\
 % &=\frac{L^{+,R}_x}{re^{g_0(b_x,b_{x+1})}}+o(R).
\end{align*}
Then for any $\alpha\in\{A,T,C,G\}^M$, $\Pp$-almost surely the information $I(\alpha)$ is equivalent to 
\begin{align}
I(\alpha)=R\sum_{x=1}^{M-1}(\beta g_0(\alpha_x,\alpha_{x+1})+e^{-\beta\left(g_0(\alpha_x,\alpha_{x+1})-g_0(b_x,b_{x+1})\right)})\frac{1}{\bar{p}_x(b)}+\epsilon_x(R). \label{eqI}
\end{align}
As for any real number $c$, the function $x\rightarrow x+e^{-x+c}$ is minimal iff $x=c$, the sum is minimal if and only if for each $x\in\{1,\cdots,M-1\}$, $$g_0(\alpha_x,\alpha_{x+1})=g_0(b_x,b_{x+1}).$$ So by Hypothesis \eqref{Hyp1} and the equality $\alpha_1=b_1$, $\Pp$ almost surely for $R$ large enough, $I(\alpha)$ is minimal iff $\alpha=b$.
\end{proof}

\subsection{Control of the estimation for the continuous time case}
In this part, %we still assume the local injectivity of $g_0$ (\ref{Hyp1}). T
we show that the probability to make at least one mistake using the global estimator $\tilde{b}$ decreases exponentially. 
\begin{Cor}\label{cornberr} Let $n_e$ be the number of wrong predictions, then $\Pp$-almost surely,  
\begin{align*} 
\lim_{R \rightarrow + \infty }\frac{-\log \Pp\left(n_e\geq 1|Z^R, b_1
\right)}{R}&\geq \Delta{F}^-.%(\Delta{F}^++\Delta{F}^-) 
\end{align*}
%and to make at least $h$ mistakes in a row:
%\begin{align*}
%&-\log \Pp\left(\text{Make at least $h$ mistakes in a row}\ |Y^R\right)\\
%&\geq KM_h^R+o(R)
%&\geq KhR+o(R)
%\end{align*}
\end{Cor}
%Fix an $x\in\{1,\cdots,M-1\}$. The probability to be right at site $x$ when you know the environment around the point $x$ is
\begin{proof}
From the first part of Theorem \ref{estiglob} we know that 
\begin{equation*}
 \Pp\left(n_e \geq 1 \big|Z^R, b_1\right)  =1-  \Pp\left(b=\tilde{b}\big|Z^R, b_1\right)= 1-\frac{1}{1+\sum_{{\alpha}\neq \tilde{b}}e^{-(I(\alpha)-I(\tilde{b}))}}.
%   =\Pp\left(b=\alpha\big|L^{+,R},\ S^R\right) 
\end{equation*}
The second part of Theorem \ref{estiglob} together with Equation \eqref{eqI} give that $\Pp$-almost surely for $R$ large enough
\begin{align*}
I(\alpha)-I(\tilde{b}) 
=\sum_{x=1}^{M-1}L^{+,R}_xF\big( g_0(\tilde{b}_x,\tilde{b}_{x+1})- g_0(\alpha_x,\alpha_{x+1}\big))+\epsilon_x(R),
\end{align*}
recall that, for $u \in \R$, $F(u)=e^{\beta u}-1-\beta u$ and $\epsilon_x(R)$ is defined in Proposition \ref{quipu1}.
Notice that we also have 
\begin{align*}
I(\alpha)-I(\tilde{b}) 
&=\sum_{x=1}^{M-1}\frac{R}{\bar{p}_x(\tilde{b})}F\big( g_0(\tilde{b}_x,\tilde{b}_{x+1})- g_0(\alpha_x,\alpha_{x+1}\big))+\epsilon_x(R) \\
&\geq R\sum_{x=1}^{M-1}F\big(g_0(\tilde{b}_x,\tilde{b}_{x+1})- g_0(\alpha_x,\alpha_{x+1}\big))+\epsilon_x(R).
\end{align*}
%As $F$ is minimal in $0$, the main contributions in the sum $\sum_{{\alpha}\neq \tilde{b}}$ come from the $\alpha$ which differ  from $\tilde{b}$ only in one site. In that case, if  $\alpha_y\neq \tilde{b}_y$ for some $1\leq y \leq M-1$ we have
For $\alpha\neq \tilde{b}$, denote by $\alpha_y$ the first site such that $\alpha_y\neq \tilde{b}_y$, obviously $y\geq 2$ and therefore,
\begin{align*}
I(\alpha)-I(\tilde{b}) 
&\geq R\cdot F\left(g_0(\tilde{b}_{y-1},\tilde{b}_{y})- g_0(\tilde{b}_{y-1},\alpha_{y})\right)+O(\sqrt{R\log\log R})\\%\left(F\left(g_0(\tilde{b}_{y-1},\tilde{b}_{y})- g_0(\tilde{b}_{y-1},\alpha_{y})\right)\right. \\ 
%&+ \left. F\left(g_0(\tilde{b}_x,\tilde{b}_{x+1})- g_0(\alpha_{y},\alpha_{y+1})\right)\right)
&\geq R\cdot \Delta F^-+O(\sqrt{R\log\log R})
\end{align*} so finally we get 
 \begin{align*}
\sum_{{\alpha}\neq \tilde{b}}e^{-(I(\alpha)-I(\tilde{b}))} \leq 4^Me^{-R\cdot \Delta{F}^-+O(\sqrt{R\log\log R})} %\leq 4^Me^{-R (\Delta{F}^++\Delta{F}^-)+O(\sqrt{R\log\log R})} 
\end{align*}
which concludes the proof.
\end{proof}
Now let us define $\tilde{n}_e$ the number of non successive errors;  by non successive, we mean that two errors are separated by at least one good prediction. The probability to make more than $h$ non successive errors exponentially decreases: %DIRE POURQUOI C EST INTERESSANT. 
%That is to say, $\tilde{n}_e$ is the number of connected compounds of the random set of the sites where the prediction is wrong.
\begin{Cor}\label{cornberr}
Let $h\in\N$, $\Pp$-almost surely,  
\begin{align*} 
\lim_{R \rightarrow + \infty }\frac{-\log \Pp\left(\tilde{n}_e\geq h|Z^R, b_1
\right)}{R}&\geq h\cdot\Delta{F}^-.%(h+1)(\Delta{F}^++\Delta{F}^-).
\end{align*}
%and to make at least $h$ mistakes in a row:
%\begin{align*}
%&-\log \Pp\left(\text{Make at least $h$ mistakes in a row}\ |Y^R\right)\\
%&\geq KM_h^R+o(R)
%&\geq KhR+o(R)
%\end{align*}
% where 
% $$
% M_h^R=\min\left(\sum_{x\in I_h}L^{+,R}_x\ /\ I_h\subset \{1,\dots,M-1\}\ ,\ |I_h|=h\right).
% $$
\end{Cor}

\begin{proof} Like in the previous section we easily compute
\begin{align*}
 \Pp\left(\tilde{n}_e \geq h \big|Z^R, b_1\right)& =\Pp \left(b \in \mathcal{A}_h \big|Z^R, b_1\right)\\ 
 & = \left(1+\frac{\sum_{{\bar{\alpha}}\in\overline{\mathcal{A}_h}}e^{-I(\bar{\alpha})}}{\sum_{{\alpha}\in\mathcal{A}_h}e^{-I(\alpha)}}\right)^{-1}.
%   =\Pp\left(b=\alpha\big|L^{+,R},\ S^R\right) 
\end{align*}
where $\mathcal{A}_h$ is the set of the chains $\alpha$ which are different from $\tilde{b}$ in at least $h$ non successive sites and $\overline{\mathcal{A}_h}$ is the complementary set. 
As $\tilde{b}\in\overline{\mathcal{A}_h}$,
\begin{equation*}
 \Pp\left(\tilde{n}_e \geq h \big|Z^R, b_1\right)\leq \left(1+\frac{e^{-I(\tilde{b})}}{\sum_{{\alpha}\in\mathcal{A}_h}e^{-I(\alpha)}}\right)^{-1},
\end{equation*}
thus, as before, we just have to study the quantities $I(\alpha)-I(\tilde{b})$ where $\alpha\in\mathcal{A}_h$.

Here the only important contribution for a given chain $\alpha$ comes from the sites $y$ such that $\alpha_y$ is different from $\tilde{b}_y$ but $\alpha_{y-1}=\tilde{b}_{y-1}$. As every chain of $\mathcal{A}_h$ has at least $h$ such points, we obtain in the same way as before,
\begin{align*}
\sum_{{\alpha}\in\mathcal{A}_h}e^{-(I(\alpha)-I(\tilde{b}))} \leq 4^Me^{-Rh\Delta F^-+O(\sqrt{R\log \log R})}, %4^Me^{-R(h+1)(\Delta F^++\Delta F^-)+O(\sqrt{R\log \log R})},
\end{align*}
it is then easy to obtain the result of the corollary.
\end{proof}

This corollary shows that we have very few chances to make several mistakes at distant bases of the molecule, however notice that we can not replace $\tilde{n}_e$ by ${n}_e$, indeed the probability to be wrong at $h$ successive sites does not decrease exponentially in $h$. We actually note that if we get a mistake in one site then there is a great probability to make mistakes on the following bases.

% because an error at one site increases the probability to make a wrong prediction at the following site. 
The discrete time case leads to very similar results, in fact the main difference is that $\Delta F^-$ is replaced by $\Delta G^-$.%$\Delta F^+$ and $\Delta F^-$ are replaced by $\Delta G^+$ and $\Delta G^-$.

%We could think that we lose lots of precision by removing the term $\frac{1}{\bar{p}_x(\tilde{b})}$.
%\subsection{Some numerical simulations}
%The results of the previous section allow to estimate the quality of the estimator.
%

%%%%%%%%%%%%%%%%%%%%%%%%%%%%%%%%%%%%%%%%%%%%%%%%%%     FIN MODIF    %%%%%%%%%%%%%%%%%%%%%%%%%%%%%%%%%%%%%%%%%%%%%%%%%%%%%%%%

\section{Possible improvements of the method\label{sec3}}

In this paragraph we use the results of the previous sections to propose two simple extensions which improve the prediction. In both of them the idea is to adapt the force $f$ to the context.
%Before going any further, we modify a little bit the definition of the $\dg$. The aim is in some sense  to be able to modify the force $f$ when we arrive at some particular coordinate of the walk or equivalently at some special pair of the DNA molecule. %We will explain later why this modification is useful and may answer a question asked by S. Cocco and R. Monasson. 
%Define 
%\begin{align}
%\dg(b_{x},b_{x+1}):=g_0(b_{x},b_{x+1})-g_1(f_x), \label{Delta}
%\end{align}
%and we assume that $g_1$ depends on a force $f$ which depends itself on the site $x$.

\subsection{ Forces depending on the coordinate of a site}
In this section, we mainly  discuss about $\frac{1}{\bar{p}_x}$ which appears in the expression  of $1/R_c(x)$. As seen before, $\frac{1}{\bar{p}_x}$  can be large but it depends on the sequence $g_0$ and the force at site $x$. We recall that
 %the $R_c$ given in (\ref{Rc1}) even if it can be applied every where when,
 \begin{align*}
\frac{1}{\bar{p}_x} %&  \sum_{l=x+1}^{M-1}\exp\left(\sum_{k=x+1}^l\beta\big(g_0(b_k,b_{k+1})-g_1(f)\big) \right)\\
=& \sum_{l=x+1}^{M-1}\exp\left(-\beta\left\{(l-x)g_1(f)-\sum_{k=x+1}^lg_0(b_k,b_{k+1})\right\}  \right).
\end{align*}
For example when the force $f$ is not too large, some valleys, that is to say portions of the sequence $b$ such that $\sum_{k=x+1}^lg_0(b_k,b_{k+1})-(l-x)g_1(f)$ is large for a given $l$, can appear.
So the quality of the prediction is good only in some specific regions of the molecule (the decrease of the probability to be wrong behaves like $e^{-\textrm{const} R e^{\beta M_x}}$, where $M_x$ is given in (\ref{Mx})).

When the above condition does not appear, the force can be modified in order to slow down locally the system\footnote{According to the physicists, it is possible. }. 
Assume that we are interested in a specific region centered at the coordinate $y$, $[y-A,y+A]$ for some $A>0$, where the $(L_{y+x}^++L_{y+x}^-, x\in [-A,A])$ or the $(S_{y+x},\  x\in [-A,A])$  are small. Then we can take for $x\in [-A,A]$,
\begin{align}
g_1(f_{y+x})=C(A-x),
\end{align}
for some small constant $C>0$, we get:
 \begin{align*}
\frac{1}{\bar{p}_{y+x}} %&\geq \sum_{l=y+x+1}^{A+y}\exp\left(\beta\left\{\sum_{k=y+x+1}^lg_0(b_k,b_{k+1})-\frac{C}{2}(l-x)(2A-l-x-1) \right\}  \right) \\
&\geq\exp\left(-\beta\left\{\frac{C}{2}(A-x)(A-x-1) -\sum_{k=y+x+1}^{A+y}g_0(b_k,b_{k+1})\right\}  \right)
\end{align*}
especially for $x=0$,
 \begin{align*}
\frac{1}{\bar{p}_{y}}&\geq\exp\left(-\beta\left\{\frac{C}{2}A(A-1)-\sum_{k=y+1}^{A+y}g_0(b_k,b_{k+1})\right\}  \right)
\end{align*}
Then if $\Ep(g_0(b_k,b_{k+1}))> {C}(A-1)$ and $A$ is large enough, $\frac{1}{\bar{p}_{y}}$ will be quite large too.
Once again that will work if the region we are looking at is quite far from the end of the molecule, that is to say, $A$ is large. On the other case what could be a good idea is to unzip the molecule from the end.
We now move to another possible improvement.

\subsection{ The energy point of view: forces depending on the values of the environment}

In this paragraph we do not try to find directly the sequence of bases but the associated binding energies. We denote $g_0(x)$ for $g_0(b_x,b_{x+1})$ and we assume that there are $K$ distinct values for $g_0(x)$, typically for a DNA molecule they are given by Table \ref{tab1}. Note that the random variables $(g_0(x),x)$ are not independent, and that the dependence is also given by Table \ref{tab1}. For example, the energy 1.06 can only be followed by 1.78, 1.55, 2.52 or 2.22. 
We will keep the notation $\Pp^{b}$ when we work at fixed energy. 
\emph{To simplify the computations we also assume that the sequences $g_0$ are equiprobable.}

First let us introduce some new notations. We will denote by $\mu_1,\mu_2,\cdots,\mu_K$, the possible values of $g_0(.)$,  ordered in such a way that $\mu_i>\mu_{i+1}$ for all $i$. We also assume that the force $f$ %can vary 
can take $K+1$ decreasing values $\{f_1,f_2,\cdots,f_{K-1},f_K,f_{K+1}=0\}$ such that $g_1(.)$ takes  $K+1$ distinct values denoted $\{r_1,r_2,\cdots,r_{K-1},r_K,r_{K+1}=0\}$ and satisfying %$ r_1>r_2>\cdots>r_{K-1}$, $\mu_i-r_j<0$ if $j\leq i$ and $\mu_i-r_j>0$ if $j> i$.
 \begin{align*}
 &  r_1>r_2>\cdots>r_{K-1}, \\
% &\mu_i-r_j<0 \text{ if } j\leq i,\\
% &\mu_i-r_j>0 \text{ if } j> i.
& \mu_1-r_1<0,\ \mu_1-r_2>0,\ \forall i>1\ \mu_i-r_2<0,  \\
&  \mu_2-r_2<0,\ \mu_2-r_3>0,\ \forall i>2\ \mu_i-r_3<0, \\
 &   \cdots  \\
&  \mu_{K}-r_{K}<0,\  \mu_{K}-r_{K+1}=\mu_{K}>0.
 \end{align*}

Let us define $q^i_m:=(1+e^{\beta(\mu_m-r_i)})^{-1}$.
% \begin{align}
% &  q^i_m:=(1+e^{\beta(\mu_m-r_i)})^{-1},
% \end{align}
This is the probability to go on the right if the force $f_i$ is applied and if the value of the environment is equal to $\mu_m$. Notice that if $f_1$ is applied then for all $x\leq M$, $p_x:=\left(1+\exp(\beta (g_0(x)-r_1)) \right)^{-1} \geq q_{1}^1 >1/2$. We denote $\Gamma_1:=\{x \leq M, g_0(x)=\mu_1 \}$. Then if $f_2$ is applied, for all $x \leq M,\ x \notin \Gamma_1$, $p_x \geq q_2^2$. We therefore denote $\Gamma_i:=\{x \leq M, g_0(x)=\mu_i \}$ for $i\in\{2,\dots,K\}$ and we get a partition $\{\Gamma_1,\Gamma_2,\cdots,\Gamma_K\}$ of $\{1,\cdots,M\}$.
The idea is then to consider a certain number of random walks for each values taken by the force.  We denote by $R_j$ the number of random walks we consider for the force with value $f_j$. \\ From now on we will only focus on the discrete time case, indeed it is the one where the gain is the most important, however what we suggest can be applied to the continuous time model as well. We introduce the information at site $x$ when the force $f_j$ is applied:
\begin{align}
& i^j_x(m):=L_x^{+,R_j}\log\left(1+e^{\beta (\mu_m-r_j})\right)+L_x^{-,R_j}\log\left(1+e^{-\beta (\mu_m-r_j)}\right) \label{info3} 
\end{align}
and the relative information at site $x$
\begin{align}
i^j_x(m,l):=i^j_x(m)-i^j_x(l).
\end{align}
We also define the function $H_a:\R \rightarrow \R_+$, 
$$ H_a(u):=\log(1+e^{\beta u})+\exp(\beta  a)\log(1+e^{-\beta u}).$$ 
\begin{proposition} Let $k \leq K$, assume that the forces $f_k$ and then $f_{k+1}$ are applied (everywhere) then, for any $x$, any sequence $g_0^x$ and any estimator $ \hat{g}_0(x)$,
\begin{align*} 
 & \Pp \left(\hat{g}_0(x)=\mu_k,\ {g_0}(x)\neq \mu_{k} |X^{R_k},X^{R_{k+1}},g_0^x \right) \\ 
 & = \left(1+\left(\sum_{m=1,m\neq k}^K \exp\left(-i^{k}_x(m,k)-i^{k+1}_x(m,k) \right) \right)^{-1} \right)^{-1}\un_{\hat{g}_0(x)=\mu_k}. %\\
 \end{align*}
Let us define the following estimator:
\begin{align}
 \hat{g}_0(x)= \mu_{\inf \left\{k>0, \ \frac{L_x^{-,R_k}}{L_x^{+,R_k}}<1,\ \frac{L_x^{-,R_{k+1}}}{L_x^{+,R_{k+1}}}>1  \right\}}, \label{estm2}
\end{align} 
then $\Pp$-almost surely,
\begin{align*} 
\frac{1}{R_c^k(x)} & := -\lim_{R_k=R_{k+1}=R \rightarrow \infty } \frac{1}{R} \log\Pp\left(\hat{g}_0(x)=\mu_k,\ {g_0}(x)\neq \mu_{k} |X^{R_k},X^{R_{k+1}},\ g_0^x \right)\\  
& \geq \frac{H^{(k)}}{\bar{p}_x^k}   + \frac{H^{(k+1)}}{\bar{p}_x^{k+1}},
\end{align*}
where
\begin{align} 
H^{(k)}&:=\min_{l \in \{k-1,k+1\}}H_{\mu_k-r_k}(\mu_l-r_k)-H_{\mu_k-r_k}(\mu_k-r_k), \nonumber \\
H^{(k+1)}&:= \min_{l \in \{k-1,k+1\}}H_{\mu_k-r_{k+1}}(\mu_l-r_{k+1})-H_{\mu_k-r_{k+1}}(\mu_k-r_{k+1}), \nonumber \textrm{ and}\\
\frac{1}{\bar{p}_{x}^l} &=\frac{1}{\bar{p}_{x}^l(g_0^x,f_l)}:=\sum_{z=x+1}^{M-1}\exp\left(\sum_{y=x+1}^z g_0(y)-r_l\right)+1.  \label{barp2}
 \end{align} 
%$\gamma^x$ is (like for $b^x$) the sequence $(\gamma(1),\gamma(2), \cdots, \gamma(x-1), \gamma(x+1),\cdots, $ $ \gamma(M))$. 
\end{proposition} 

We first give a short proof of the result and then discuss about the improvement.
\begin{proof}
The first part of the proposition is, like before, easily deduced from Bayes formula. Thanks to Lemma \ref{loiLp} and the LLN, $\Pp^{b}$-almost surely
\begin{align*}
 \lim_{R_k \rightarrow + \infty} \frac{i^{k}_x(m,k)}{R_k} %& =\frac{1}{\bar{p}_x^k}\left(\log \frac{q_m^k}{q_k^k}+e^{\beta(g_0(x)-r_k)}\log \frac{1-q_m^k}{1-q_k^k}\right) \\ 
& = \frac{1}{\bar{p}_x^k}\left( H_{g_0(x)-r_k}(\mu_m-r_k)-H_{g_0(x)-r_k}(\mu_k-r_k) \right),
\end{align*}
and in the same way  $\Pp^{b}$-almost surely
\begin{align}
& \lim_{R_j=R\rightarrow + \infty,\ \forall 1\leq j \leq K} \inf \left\{k>0, \ \frac{L_x^{-,R_k}}{L_x^{+,R_k}}<1,\ \frac{L_x^{-,R_{k+1}}}{L_x^{+,R_{k+1}}}>1  \right\}  \nonumber \\
&=\inf \left\{k>0,\ g_0(x)-r_k<0,\ g_0(x)-r_{k+1}>0  \right\}. 
\end{align}
This implies the  $\Pp^{b}$-almost sure convergence of $\{ \hat{g}_0(x)=\mu_k \}$ to the event $\{g_0(x)=\mu_k\}$. Therefore as $H_{a}(.)$ gets its minimum in $a$, $\Pp^{b}$ almost surely on $\{ \hat{g}_0(x)=\mu_k \}$ for all $m \neq k$
\begin{align*}
\lim_{R_k \rightarrow + \infty} \frac{i^{k}_x(m,k)}{R_k} %&  \leq  -\frac{1}{\bar{p}_x^k}\min_{l\neq k}\left(H_{g_0(x)-r_k}(\mu_k-r_k)- H_{g_0(x)-r_k}(\mu_l-r_k) \right), \\
& = \frac{1}{\bar{p}_x^k} \left(H_{\mu_k-r_k}(\mu_m-r_k)- H_{\mu_k-r_k}(\mu_k-r_k) \right), \\
& =: \frac{1}{\bar{p}_x^k} \Delta H_k^{(k)}(m)>0.
\end{align*}
A similar analysis can be done for $i^{k+1}_x(m,k)$ so $\Pp^{b}$-almost surely  $$\lim_{R_{k+1} \rightarrow + \infty} \frac{i^{k+1}_x(m,k)}{R_{k+1}}=:\frac{1}{\bar{p}_x^{k+1}} \Delta H_k^{(k+1)}(m).$$ We get that $\Pp^{b}$-almost surely for $R_k=R_{k+1}=R$ large enough
\begin{align*} 
 & \log \Pp\left(\hat{g}_0(x)=\mu_k,\ {g_0}(x)\neq \mu_{k} |X^{R_k},X^{R_{k+1}},\ g_0^x \right) \\ 
% &  \leq - \log \left(1+\left(\sum_{m=1,m\neq k}^K \exp\left(-\frac{R_{k}}{\bar{p}_x^k} \Delta H_k^{(k)}(m)-\frac{R_{k+1}}{\bar{p}_x^{k+1}} \Delta H_k^{(k+1)}(m)+o(R_k) \right) \right)^{-1} \right) \\
&  \leq \log \left(\sum_{m=1,m\neq k}^K \exp\left(-\frac{R}{\bar{p}_x^k} \Delta H_k^{(k)}(m)-\frac{R}{\bar{p}_x^{k+1}} \Delta H_k^{(k+1)}(m)+o(R) \right) \right)  \\
 &\leq -\min_{m\neq k} \left\{\frac{R}{\bar{p}_x^k} \Delta H_k^{(k)}(m)+\frac{R}{\bar{p}_x^{k+1}} \Delta H_k^{(k+1)}(m) +o(R)\right\}.
 \end{align*}
The $o(R_k)$ is the negligible term that comes from the ILL (see the end of the proof of Proposition \ref{quipu1}) which is of order of $\sqrt{R \log \log R}$. This gives the desire result by dividing by $R$. %\textbf{remplacer les $R_k$ par des $R$}
\end{proof}
Here we avoid a bad situation seen in the first section (see (\ref{M10})): for large $\beta$
\begin{align*} 
%H^{(k+1)}\approx \beta (\mu_k-r_{k+1}) \exp(\beta(\mu_k-r_{k+1}))-1-\beta(\mu_k-r_{k+1}) 
H^{(k+1)}\approx \beta \exp(\beta(\mu_k-r_{k+1}))
\end{align*} 
%if we only focus on the contribution given by the force $f_{k+1}$ gives :
which exponentially increases with $\beta$.
However we have to be careful with this method. In order to catch the small values of the energy, $f_k$ should be small and may slow down the system(see (\ref{time2r}) and Lemma \ref{loiLp}), indeed we have  
\begin{align*}
\Ep^b\left[\tau_M\right] =&R \sum_{x=1}^{M-2}\left(\frac{1}{\bar{p}_x^k}+\frac{1}{\bar{p}_{x-1}^k}-1\right)+ R \sum_{x=1}^{M-2}\left(\frac{1}{\bar{p}_x^{k+1}}+\frac{1}{\bar{p}_{x-1}^{k+1}}-1\right) \\
\geq & R\exp(\beta \max _x M_x^{k+1}),
\end{align*}
and $M_x^{k+1} := \max_{ x \leq l\leq M-2}\left\{\sum_{l=x+1}^lg_0(b_l,b_{l+1})-r_{k+1} \right\}$ is large if $k$ is close to $K$. \\
An alternative approach is first to apply a large force $f_1$ from $0$ to $x-1$ in order to reach quickly the region we are interested in, then to apply all the forces in $x$ and then, after $x+1$ to apply a small force (for example $f_{K}$) in order to slow down the system and stay focus on $x$. More precisely $f$ depends on the energie as before but it also depends on the site: 
\begin{align*}
f_i(z)=f_1\un_{1 \leq z \leq x-1}+f_i\un_{ z = x}+f_{K}\un_{ z \geq x+1}.
\end{align*}
We get the following $\Pp$-almost sure result
\begin{align} 
\frac{1}{R_c(x)} & := - \lim_{R_j=R\rightarrow + \infty,\ \forall 1\leq j \leq K} \frac{1}{R} \log \Pp\left(\hat{g}_0(x)\neq{g_0}(x)|(X^{R_i}, i \leq K), g_0^x \right) \nonumber \\  
& \geq \frac{1}{\bar{p}_x^{K}}(H^{\rightarrow}+H^{\leftarrow}), \label{theeq}  \\  
H^{\rightarrow}& := \max_{ k \leq K-1 } \min_{l \in \{k-1,k+1\}}(H_{\mu_l-r_k}(\mu_l-r_k) -H_{\mu_k-r_k}(\mu_k-r_k)), \nonumber\\
H^{\leftarrow}& :=\max_{ k \leq K-1} \min_{l \in \{k-1,k+1\}}(H_{\mu_l-r_k}(\mu_l-r_{k+1})-H_{\mu_k-r_{k+1}}(\mu_k-r_{k+1}))\nonumber.
\end{align}
The main interest in the above result comparing to the previous one is the fact that $\frac{1}{\bar{p}_x^{K}}$ is large but the time to reach $x$ is small. Indeed  $$\Ep^b\left[\tau_x\right]  \sim   R \times x\left(1+ \exp\left(\beta  \max_{ 1 \leq l\leq x-1}\left\{\sum_{j=1}^lg_0(j)-r_1 \right\} \right) \right) \leq 2R \times x.$$ Of course this also increases the time required to reach the end of the molecule, but we can imagine that the process can be stopped once the precision for the site $x$ is reached. 
The proof to get the above expression is very close to the previous one so we do not give any details.

A last remark, the prediction depends on the rest of the unknown sequence $g_0^x$ due to the presence of ${1}/{\bar{p}_x^{K}}$. We can imagine an extreme case where  the forces  $f_i(z)=f_1\un_{1 \leq z \leq x-1}+f_i\un_{ z = x}+f_{K+1}\un_{ z \geq x+1}$ from $i=1$ to $K$ are applied, which means that the molecule can not be split after the base $x$. In this case we would have $\Pp$-almost surely for any sequence $g_0(x)$,
\begin{align} 
\frac{1}{R_c(x)} & := -\lim_{R \rightarrow \infty } \frac{1}{R} \log\left(\Pp\left(\hat{g}_0(x)\neq{g_0}(x)|X^R,\ g_0^x, f_i(.), i \leq K \right) \right)\nonumber \\  
& \geq  (H^{\rightarrow}+H^{\leftarrow}) \exp(\mu_K \beta (M-x)). \label{theeq}  
\end{align}
so at least asymptotically we get a lower bound for $ {1}/{R_c(x)}$  which is independent of $g_0^x$ and exponentially increasing in $\beta$.

We conclude with a discussion about the link between the energie and the sequence of bases. First let us recall the table of the binding free energies for DNA at room temperature:
\begin{center}
\begin{tabular}{|l|l|l|l|l|}
%\hline  $i=$ & 0 & 1 \\
\hline   $g_0$ & A & T & C & G   \\
\hline   A  &  1.78 & 1.55  & 2.52 & 2.22  \\
\hline
 T  &  1.06 & 1.78  & 2.28 & 2.54  \\
\hline
 C  &  2.54 & 2.22  & 3.14 & 3.85  \\
\hline
 G  &  2.28 & 2.52  & 3.90 & 3.14  \\
\hline
\end{tabular}
\end{center}
Notice that the largest free energies which correspond to the most stable links are on the bottom right end corner of the table, in fact the largest binding energy is obtained when a $G$ is followed by a $C$. Notice also that $g_0(G,G)=g_0(C,C)$ so we can not distinguish these two different links by looking only at the free energy. In the same way the lowest free energy is produced by bases $T$ and $A$ followed by the same letters, once again $g_0(A,A)=g_0(T,T)$. For the rest of the table we have the equality $g_0(W,S)=g_0(\bar{S},\bar{W})$, where $S$ is either a $C$ or a $G$ and $W$ a $A$ or a $T$, $\bar{S}$ (respectively $\bar{W}$) is the complementary of $S$ (respectively of $W$).

Moreover it is possible to reconstruct the DNA molecule from the compatible binding energies only if there is only one sequence of base pairs which corresponds to the sequence of energies (see Theorem \ref{estiglob}). This is not always the case, for example when the molecule repeats  the same scheme: the energy of $C-C-\cdots-C$  is equal to the energy of $G-G-\cdots-G$, in the same way $A-C-A-C-\cdots -A-C$ has the same energy than $G-T-G-T-\cdots -G-T$. Notice that if these highly improbable sequences are broken only once in the molecule then we turn back to a solvable case.%, indeed for example: $A-C-A-C-A-C-B$ and $A-C-A-C-A-C-\bar{B}$ have a different free energy sequence.

\vspace{0.5cm}
\noindent \textbf{Acknowledgments}
We would like to thank Nathanael Enriquez and the members of the ANR MEMEMO  who enable us to meet R\'emi Monasson. Also we would like to thank R\'emi Monasson and Simona Cocco for introducing the subject, sharing several discussions and for a kind invitation at the ENS. %Finally we would like to thank R. Abraham for useful discussions.

\bibliographystyle{plain} 
% \bibliography{thbiblio} 

\begin{thebibliography}{}

\end{thebibliography}


\begin{thebibliography}{10}

\bibitem{AdeEnr}
O.~Adelman and N.~Enriquez.
\newblock Random walks in random environment: What a single trajectory tells.
\newblock {\em Israel J. Math.}, 142:205--220, 2004.

\bibitem{Pierre7}
P.~Andreoletti.
\newblock On the estimation of the potential of {S}inai's rwre.
\newblock {\em Braz. J. Probab. Stat.}, 25:121-144, 2011.

\bibitem{Monasson0}
V.~Baldazzi, S.~Cocco, E.~Marinari, and R.~Monasson.
\newblock Infering dna sequences from mechanical unzipping: an ideal-case
  study.
\newblock {\em Physical Review Letters E}, \textbf{96}:\ 128102--1--4, 2006.

\bibitem{Monasson1}
V.~Baldazzi, S.~Cocco, E.~Marinari, and R.~Monasson.
\newblock Infering dna sequences from mechanical unzipping data: the
  large-bandwith case.
\newblock {\em Physical Review Letters E}, \textbf{75}:\ 011904--1--33, 2007.

\bibitem{Bockelmann}
U.~Bockelmann, B.~Essevaz-Roulet, and F.~Heslot.
\newblock Molecular stick-slip motion revealed by opening dna with piconewton
  forces.
\newblock {\em Phys. Rev. Let.}, \textbf{79}:\ 4489--4492, 1997.

\bibitem{Bockelmann2}
U.~Bockelmann, B.~Essevaz-Roulet, and F.~Heslot.
\newblock Dna strand separation studied by single molecule force measurements.
\newblock {\em Phys. Rev. E}, \textbf{58}:\ 2386--2394, 1998.

\bibitem{Monasson2}
S.~Cocco and R.~Monasson.
\newblock Reconstructing a random potential from its random walks.
\newblock {\em epl}, \textbf{81}:\ 1--6, 2008.

\bibitem{KesKozSpi}
H.~Kesten, M.V. Kozlov, and F.~Spitzer.
\newblock A limit law for random walk in a random environment.
\newblock {\em Comp. Math.}, \textbf{30}:\ 145--168, 1975.

\bibitem{Sinai}
Ya.~G. Sinai.
\newblock The limit behaviour of a one-dimensional random walk in a random
  medium.
\newblock {\em Theory Probab. Appl.}, \textbf{27}(2):\ 256--268, 1982.

\bibitem{Solomon}
F.~Solomon.
\newblock Random walks in random environment.
\newblock {\em Ann. Probab.}, \textbf{3}(1):\ 1--31, 1975.

\bibitem{Viterbi}
A.~J. Viterbi.
\newblock Error bounds for convolutional codes and an asymptotically optimum
  decoding algorithm.
\newblock {\em IEEE Trans. Inf. Theory}, \textbf{13}(2):260--269, 1967.

\end{thebibliography}
%\newpage

\end{document}